\documentclass[noderivs,copyright]{eptcs}

\providecommand{\event}{AFL 2017} 

\usepackage{underscore}           

\usepackage{graphicx} 
\usepackage{mathrsfs}
\usepackage{amssymb}
\usepackage{amsmath}
\usepackage{theorem}
\usepackage{multicol}

\newcommand{\lfam}{\mathscr{L}}

\newcommand{\dollar}{\texttt{\$}}

\newcommand{\border}{\texttt{\#}}
\newcommand{\valc}{\textrm{VALC}}
\newcommand{\invalc}{\textrm{INVALC}}

\newcommand{\iddhpda}{\textsf{ID2hPDA}}
\newcommand{\dhpda}{\textsf{2hPDA}}

\newcommand{\ndhpda}{\textsf{ndet-2hPDA}}

\newcommand{\doiddhpda}{\textsf{double-ID2hPDA}}
\newcommand{\ddoiddhpda}{\textsf{det-}\doiddhpda}
\newcommand{\ndoiddhpda}{\textsf{ndet-}\doiddhpda}

\newcommand{\diddhpda}{\textsf{det-}\iddhpda}
\newcommand{\niddhpda}{\textsf{ndet-}\iddhpda}

\newcommand{\npda}{\textsf{NPDA}}
\newcommand{\dpda}{\textsf{DPDA}}

\newcommand{\squareforqed}{$\Box$}
\newcommand{\squareforeoe}{$\blacksquare$}
\newcommand{\eoe}{\ifmmode\squareforeoe\else{\unskip\nobreak\hfil
\penalty50\hskip1em\null\nobreak\hfil\squareforeoe
\parfillskip=0pt\finalhyphendemerits=0\endgraf}\fi}
\newcommand{\qed}{\ifmmode\squareforqed\else{\unskip\nobreak\hfil%
  \penalty50\hskip1em\null\nobreak\hfil\squareforqed%
  \parfillskip=0pt\finalhyphendemerits=0\endgraf}\fi}

\theorembodyfont{\slshape} 
\theoremstyle{plain} 
 \newtheorem{definition}{Definition}
 \newtheorem{lemma}[definition]{Lemma}
 \newtheorem{theorem}[definition]{Theorem}
 \newtheorem{corollary}[definition]{Corollary}

\theorembodyfont{\normalfont}
 \newtheorem{example}[definition]{Example}

\newenvironment{proof}{\noindent{\textbf{Proof}\ }}{\qed\medskip}

\title{Input-Driven Double-Head Pushdown Automata}

\author{Markus Holzer, Martin Kutrib, Andreas Malcher, Matthias Wendlandt
\institute{Institut f\"ur Informatik, Universit\"at Giessen,\\
  Arndtstr.~2, 35392 Giessen, Germany} 
\email{$\{$holzer,kutrib,malcher,matthias.wendlandt$\}$@informatik.uni-giessen.de}
}

\def\titlerunning{Input-Driven Double-Head Pushdown Automata}
\def\authorrunning{M.~Holzer, M.~Kutrib, A.~Malcher \& M.~Wendlandt}

\begin{document}

\maketitle

\begin{abstract}
  We introduce and study input-driven deterministic and
  nondeterministic double-head pushdown automata. A double-head
  pushdown automaton is a slight generalization of an ordinary
  pushdown automaton working with two input heads that move in
  opposite directions on the common input tape. In every step one head
  is moved and the automaton decides on acceptance if the heads
  meet. Demanding the automaton to work input-driven it is required
  that every input symbol uniquely defines the action on the pushdown
  store (push, pop, state change). Normally this is modeled by a
  partition of the input alphabet and is called a
  \emph{signature}. Since our automaton model works with two heads
  either both heads respect the same signature or each head owes its
  own signature. This results in two variants of input-driven
  double-head pushdown automata. The induced language families on
  input-driven double-head pushdown automata are studied from the
  perspectives of their language describing capability, their closure
  properties, and decision problems.
\end{abstract}

\section{Introduction}
\label{sec:intro}

Input-driven pushdown automata were introduced
in~\cite{mehlhorn:1980:pmradcflr} in the course of deterministic
context-free language recognition by using a pebbling strategy on the
mountain range of the pushdown store. The idea on input driven
pushdown automata is that the input letters uniquely determine whether
the automaton pushes a symbol, pops a symbol, or leaves the pushdown
unchanged. The follow-up papers~\cite{braunmuehl:1983:idlrlog} and~\cite{dymond:1988:idllognd} studied
further properties of the family of input-driven pushdown languages.
One of the most important properties on input-driven pushdown
languages is that deterministic and nondeterministic automata are
equally powerful. Moreover, the language family accepted is closed
under almost all basic operations in formal language theory. Although
the family of input-driven pushdown languages is a strict subset of
the family of deterministic context-free languages, the input-driven
pushdown languages are still powerful enough to describe important
context-free-like structures and moreover share many desirable
properties with the family of regular languages. These features turned
out to be useful in the context of program analysis and led to a
renewed interest~\cite{Alur:2009:answ} on input-driven pushdown languages
about ten years ago. 
In~\cite{Alur:2009:answ} an alternative
name for input-driven pushdown automata and languages was coined,
namely visibly pushdown automata and languages. Sometimes
input-driven pushdown languages are also called nested word
languages. Generally speaking, the revived research on input-driven
pushdown languages triggered the study of further input-driven
automata types, such as input-driven variants of, e.g., (ordered)
multi-stack automata~\cite{carotenuto:2016:omsvpda}, stack automata~\cite{bensch:2012:idsa:proc}, queue
automata~\cite{kutrib:2015:idqaftdcp}, etc.

We contribute to this list of input-driven devices, by introducing and
studying input-driven double-head pushdown automata. Double-head
pushdown automata were recently introduced
in~\cite{Nagy:2015:afothpa}.\footnote{Originally these devices were named two-head
  pushdown automata in~\cite{Nagy:2015:afothpa}, but since this naming may cause
  confusion with multi-head pushdown automata of~\cite{harrison:1968:mtmhpda}, we use
  to refer to them as double-head pushdown automata instead.}  Instead
of reading the input from left to right as usual, in a double-head
pushdown automata the input is processed from the opposite ends of the
input by double-heads, and the automaton decides on acceptance when
the two heads meet. Thus, double-head pushdown automata are a straight
forward generalization of Rosenberg's double-head finite automata for
linear context-free languages~\cite{rosenberg:1967:mrlcfl}---see
also~\cite{nagy:2012:cothfall}.  The family of double-head
nondeterministic pushdown languages is a strict superset of the family
of context-free languages and contains some linguistically important
non-context-free languages. In fact, the family of double-head
nondeterministic pushdown languages forms is a mildly
context-sensitive language family because in addition to the
aforementioned containment of important languages, the word problem of
double-head nondeterministic pushdown languages remains solvable in
deterministic polynomial time as for ordinary pushdown automata.
Moreover, every double-head nondeterministic pushdown language is
semi-linear. Double-head pushdown automata are a moderate extension of
ordinary pushdown automata because languages accepted by double-head
pushdown automata still satisfy an iteration or pumping lemma. Thus,
double-head pushdown automata and the properties of their accepted languages 
are interesting objects to study.

In the next section we introduce the necessary notations on
double-head pushdown automata and their input-driven versions.
Demanding the automaton to work input-driven it is required that every
input symbol uniquely defines the action on the pushdown store (push,
pop, state change). Normally this is modeled by a partition of the
input alphabet and is called a \emph{signature}. Since our automaton
model works with two heads either both heads respect the same
signature or each head owes its own signature. This results in
(simple) input-driven and double input-driven double-head pushdown
automata. Then in Section~\ref{sec:power} we investigate the
computational capacity of (double) input-driven double-head pushdown
automata. We show that nondeterministic machines are more powerful
than deterministic ones, for both input-driven variants. Moreover, it
turns out that the language families in question are incomparable to
classical language families such as the growing context-sensitive
languages, the Church-Rosser languages, and the context-free
languages.  As a byproduct we also separate the original language
families of double-head deterministic and double-head nondeterministic
pushdown languages. Section~\ref{sec:closure} is then devoted to the
closure properties of the families of input driven
double-head pushdown languages and finally in Section~\ref{sec:dec} we
consider decision problems for the language families in question. Here
it is worth mentioning that although some problems are already not
semidecidable even for deterministic machines, the question of whether
a given deterministic input-driven double-head pushdown automaton~$M$
is equivalent to a given regular language is decidable. In contrast, the
decidability of this question gets lost, if~$M$ is a nondeterministic
input-driven double-head pushdown machine. We have to leave
open the status of some decision problems such as equivalence and
regularity. This is subject to further research.

\section{Preliminaries}
\label{sec:prelim}

Let~$\Sigma^*$ denote the set of all words over the finite
alphabet~$\Sigma$.  The \emph{empty word} is denoted by~$\lambda$, and
$\Sigma^+ = \Sigma^* \setminus \{\lambda\}$.  
For convenience, throughout the paper we use~$\Sigma_\lambda$ for
$\Sigma\cup\{\lambda\}$.  The set of words of length at most $n\geq
0$ is denoted by $\Sigma^{\leq n}$.
The \emph{reversal} of a word~$w$ is denoted by~$w^R$.  For the
\emph{length} of~$w$ we write~$|w|$. For the number of occurrences of
a symbol~$a$ in~$w$ we use the notation~$|w|_a$.  \emph{Set inclusion} is
denoted by~$\subseteq$ and \emph{strict set inclusion} by $\subset$.  We
write~$2^{S}$ for the power set and~$|S|$ for the cardinality of a
set~$S$.

A \emph{double-head pushdown automaton} is a pushdown automaton that is equipped with 
two read-only input heads that move in opposite directions on a common input
tape. In every step one head is moved.
The automaton halts when the heads would pass each other.

A pushdown automaton is called \emph{input-driven} if the input
symbols currently read define the next action on the pushdown store,
that is, pushing a symbol onto the pushdown store, popping a symbol
from the pushdown store, or changing the state without modifying the
pushdown store. To this end, we assume the input alphabet~$\Sigma$
joined with~$\lambda$ to be partitioned into the sets $\Sigma_N$,
$\Sigma_D$, and~$\Sigma_R$, that control the actions state change only
($N$), push ($D$), and pop ($R$).

Formally, a \emph{nondeterministic input-driven double-head pushdown automaton} 
($\niddhpda$) is a system 
\mbox{$M=\langle Q,\Sigma, \Gamma, q_0, F,\bot, \delta_D, \delta_R, \delta_N\rangle$,} 
where
$Q$ is the finite set of \emph{states},
$\Sigma$ is the finite set of \emph{input symbols} partitioned into the sets 
$\Sigma_D$, $\Sigma_R$, and $\Sigma_N$,
$\Gamma$ is the finite set of \emph{pushdown symbols},
$q_0 \in Q$ is the \emph{initial state},
$F\subseteq Q$ is the set of \emph{accepting states}, 
$\bot \notin \Gamma$ is the \emph{empty pushdown symbol},
$\delta_D$ is the partial transition function mapping from
  $Q \times (\Sigma_D \cup \{\lambda\})^2 \times (\Gamma \cup\{\bot\})$ 
  to $2^{Q \times \Gamma}$,
$\delta_R$ is the partial transition function mapping from
  $Q  \times (\Sigma_R \cup \{\lambda\})^2 \times (\Gamma \cup\{\bot\})$ to~$2^Q$,
$\delta_N$ is the partial transition function mapping from
  \mbox{$Q \times (\Sigma_N \cup \{\lambda\})^2 \times (\Gamma \cup\{\bot\})$}
 to~$2^Q$,
where all transition functions are defined only if the second or third
argument is~$\lambda$, and none of the transition functions is defined for
$Q \times \{\lambda\}^2 \times (\Gamma \cup\{\bot\})$.

A \emph{configuration} of an $\niddhpda$
$M=\langle Q,\Sigma, \Gamma, q_0, F,\bot, \delta_D, \delta_R, \delta_N\rangle$
is a triple $(q,w,s)$, where $q \in Q$
is the current state, $w\in \Sigma^*$ is the unread part of the input, and
$s \in \Gamma^*$ denotes the current pushdown content, where the leftmost symbol
is at the top of the pushdown store. 
The \emph{initial configuration} for an input string $w$ is set to $(q_0, w, \lambda)$.
During the course of its computation, $M$ runs through a
sequence of configurations. One step from a configuration to its
successor configuration is denoted by~$\vdash$.
Let $a \in \Sigma$, $w \in \Sigma^*$, $z'\in\Gamma$, $s \in \Gamma^*$,
and $z=\bot$ if $s=\lambda$ and $z=z_1$ if $s=z_1s_1 \in \Gamma^+$. We set
\begin{enumerate}
\item 
  $(q,aw,s) \vdash (q',w,z's)$, 
  if $a \in \Sigma_D$ and $(q',z') \in \delta_D(q,a,\lambda,z)$,
\item 
  $(q,wa,s) \vdash (q',w,z's)$, 
  if $a \in \Sigma_D$ and $(q',z') \in \delta_D(q,\lambda,a,z)$,
\item 
  $(q,aw,s) \vdash (q',w,s')$, 
  if $a \in \Sigma_R$ and $q' \in \delta_R(q,a,\lambda,z)$,\\ where
  $s'=\lambda$ if $s=\lambda$ and $s'=s_1$ if $s=z_1s_1 \in \Gamma^+$.
\item 
  $(q,wa,s) \vdash (q',w,s')$, 
  if $a \in \Sigma_R$ and $q' \in \delta_R(q,\lambda,a,z)$,\\ where
  $s'=\lambda$ if $s=\lambda$ and $s'=s_1$ if $s=z_1s_1 \in \Gamma^+$.
\item 
  $(q,aw,s) \vdash (q',w,s)$, 
  if $a \in \Sigma_N$ and $q' \in \delta_N(q,a,\lambda,z)$,
\item 
  $(q,wa,s) \vdash (q',w,s)$, 
  if $a \in \Sigma_N$ and $q' \in \delta_N(q,\lambda,a,z)$,
\end{enumerate}
So, whenever the pushdown store is empty, the successor configuration is computed
by the transition functions with the special empty pushdown symbol $\bot$, and
at most one head is moved.
As usual, we define the reflexive and transitive closure of $\vdash$ by
$\vdash^*$.
The language accepted by the $\niddhpda$~$M$ is the set $L(M)$ of words for which
there exists some computation beginning in the initial
configuration and halting in a configuration in which the whole input
is read and an accepting state is entered. Formally:
$$
L(M) = \{\,w \in \Sigma^* \mid (q_0, w, \lambda) \vdash^*
(q, \lambda, s) \text{ with } q \in F,\\ s \in \Gamma^* \,\}.
$$

The partition of an input alphabet into the sets 
$\Sigma_D$, $\Sigma_R$, and $\Sigma_N$ is called a \emph{signature}.
We also consider input-driven double-head pushdown automata, where
each of the two heads may have its own signature. To this end,
we provide the signatures $\Sigma_{D,l}$, $\Sigma_{R,l}$, and $\Sigma_{N,l}$
as well as $\Sigma_{D,r}$, $\Sigma_{R,r}$, and $\Sigma_{N,r}$ and
require for \emph{double input-driven double-head pushdown automata}
($\doiddhpda$) that they obey the first signature whenever the
left head is moved and the second signature whenever the right
head is moved.
 
If there is at most one choice of action for any possible configuration, 
we call the given (double) input-driven double-head pushdown automaton
\emph{deterministic} ($\diddhpda$ or $\ddoiddhpda$). 

In general, the family of all languages accepted by an automaton of 
some type~$X$ will be denoted by~$\lfam(X)$.

In order to clarify this notion we continue with an example.

\begin{example}\label{exa:gladkij}
The Gladkij language $\{\,w\border w^R \border w \mid w \in \{a,b\}^*\,\}$
is not growing context sensitive~\cite{Buntrock:1998:gcslcrl:art}
and, thus, is neither context free nor
Church-Rosser~\cite{McNaughton:1988:crtsfl}.
The same is true for its marked variant
\[
L_1=\{\,w\border h_1(w)^R \border h_2(w) \mid w \in \{a,b\}^*\,\},
\]
where the homomorphisms $h_1$ and $h_2$ are defined by
$h_1(a)=\bar{a}$, $h_1(b)=\bar{b}$, $h_2(a)=\hat{a}$, and $h_2(b)=\hat{b}$.
However, language $L_1$ is accepted by the $\diddhpda$ 
\[
M=\langle
\{q_0,q_1,q_a,q_b,q_+\},\Sigma_D\cup\Sigma_R\cup\Sigma_N,\{A,B\},q_0,
 \{q_+\}, \bot, \delta_D, \delta_R, \delta_N \rangle,
\]
where $\Sigma_D=\{a,b\}$, $\Sigma_R=\{\bar{a},\bar{b}\}$,
$\Sigma_N=\{\border,\hat{a},\hat{b}\}$, and
the transition functions are defined as follows. Let $X\in\{A,B,\bot\}$.
$$
\begin{array}[t]{rrcl}
 (1) & \delta_D(q_0,a,\lambda,X) &=& (q_0, A)\\
 (2) & \delta_D(q_0,b,\lambda,X) &=& (q_0, B)\\[2mm]
 (3) & \delta_R(q_1,\bar{a},\lambda,A) &=& q_a\\
 (4) & \delta_R(q_1,\bar{b},\lambda,B) &=& q_b\\
\end{array}
\qquad\qquad
\begin{array}[t]{rrcl}
 (5) & \delta_N(q_0,\border,\lambda,X) &=& q_1\\
 (6) & \delta_N(q_a,\lambda,\hat{a},X) &=& q_1\\
 (7) & \delta_N(q_b,\lambda,\hat{b},X) &=& q_1\\
 (8) & \delta_N(q_1,\border,\lambda,\bot) &=& q_+\\
\end{array}
$$
The idea of the construction is as follows. 
In a first phase, $M$ reads and pushes the input prefix 
$w$ (Transitions~1 and~2). On reading the left symbol~$\border$
automaton $M$ enters state $q_1$ which is used in the second phase.
Basically, in the second phase the left and right head are moved
alternately. When the left head is moved, the input symbol read is
compared with the symbol on the top of the pushdown store. If both
coincide, state $q_a$ or $q_b$ is entered to indicate that the right
head has to read symbol $\hat{a}$ or $\hat{b}$ (Transitions~3 and 4).
If the right head
finds the correct symbol, state $q_1$ is entered again (Transitions 6 and 7).
The second phase ends when the left head reads the second symbol~$\border$.
In that case state $q_+$ is entered and $M$ halts
(Transition 8). If in this situation the input has been read entirely and the
pushdown store is empty, clearly, the $w$ pushed in the
first phase has successfully be compared with the factor
$h_1(w)^R$ and the suffix $h_2(w)$. So, the input belongs to $L_1$.
In any other case, $M$ halts without entering the sole accepting state $q_+$.
\eoe
\end{example}

\section{Computational capacity}
\label{sec:power}

In order to explore the computational capacity of input-driven double-head
pushdown automata we first turn to show that nondeterminism is better
than determinism. As witness language for that result we use the
language $L_{dta}=\{\, a^nb^nc^n \dollar_l c^ib^ja^k \mid i,j,k,n\geq 0\,\}
\cup \{\, a^ib^jc^k \dollar_r c^nb^na^n \mid i,j,k,n\geq 0\,\}$.

\begin{lemma}\label{lem:lang-det-not}
The language $L_{dta}$ is not accepted by any deterministic double-head
pushdown automaton.
\end{lemma}

\begin{proof}
In contrast to the assertion assume that $L_{dta}$ is accepted by
some deterministic double-head pushdown automaton~$M$.
For all $m,n\geq 0$ and $x\in\{l,r\}$ we consider the input words
$a^mb^mc^m \dollar_x c^nb^na^n$ that belong to $L_{dta}$.
Since $M$ is deterministic, the computations on the words
$a^mb^mc^m \dollar_l c^nb^na^n$ and $a^mb^mc^m \dollar_r c^nb^na^n$
are identical until one of the heads reaches the center marker~$\dollar_x$.
So, we can define the set
$$
R=\{\, (m,n)\mid \text{on input } a^mb^mc^m \dollar_x c^nb^na^n
\text{ the right head of } M \text{ reaches } \dollar_x \text{ not after the
  left head}\,\}.
$$
Thus, the initial part of such a computation is in the form
$(q_0, a^mb^mc^m \dollar_x c^nb^na^n, \lambda)
\vdash^* (q_1, u\dollar_x,z)
\vdash (q_2, u,z')$,
where $q_1,q_2\in Q$, $z,z'\in\Gamma^*$, $u$ is a suffix of
$a^mb^mc^m$, and the last transition applied is of the form
$\delta(q_1,\lambda, \dollar_x,z_1)$. That is, in the last step the right head of~$M$
reads~$\dollar_x$ while seeing $z_1$ on top of the pushdown store, and the pushdown
store content~$z$ is replaced by~$z'$.
Next, the set $R$ is further refined into
$R(m)= \{\, n\mid (m,n)\in R\,\}$, for all $m\geq 0$.
Clearly, we have $R=\bigcup_{m\geq 0} (m,R(m))$.

Now assume that there is an $m\geq 0$ such that $|R(m)|$ is infinite.
We sketch the construction of a deterministic pushdown automaton~$M_1$ that
accepts the language $L_1=\{\,a^jb^jc^j\dollar_r \mid j\in R(m)\,\}$ as follows.
On a given input $a^jb^jc^j\dollar_r$, $M_1$ basically simulates a computation of
$M$ on input $a^mb^mc^m \dollar_r c^jb^ja^j$. Since $m$ is fixed, $M_1$
handles the prefix $a^mb^mc^m$ in its finite control. Moreover, since the 
left head of $M$ reaches the center marker not before the right head,
$M_1$ handles the left head and its moves in the finite control as well.
So, whenever $M$ moves its right head to the left, $M_1$ moves its sole
head to the right. Let~$[u, q]$ denote a state of $M_1$ that says that
$q$ is the simulated state of $M$ and $u$ is the still unprocessed
suffix of the prefix $a^mb^mc^m$. Then the simulation of $M$ is
straightforward: If $M$ performs a computation
$(q_0, a^mb^mc^m \dollar_r c^jb^ja^j, \lambda) 
\vdash (q_1,u\dollar_rv,  z)$, where $u$ is a suffix of
 $a^mb^mc^m$ and $v$ is a prefix of $c^jb^ja^j$, then~$M_1$
performs a computation
$([a^mb^mc^m,q_0], c^jb^ja^j\dollar_r, \lambda) 
\vdash ([u,q_1], (\dollar_rv)^R,  z)$.
When $M_1$ has read the symbol~$\dollar_r$, it continues the simulation
of $M$ with $\lambda$-steps, where now all head movements are handled in the
finite control. Finally, $M_1$ accepts if and only if $M$ accepts.
So, since $M_1$ is a deterministic pushdown automaton, $L_1$ must
be a context-free language. However, since $|R(m)|$ is assumed to
be infinite, language $L_1$ is infinite. A simple application of the
pumping lemma for context-free languages shows that any
infinite subset of $\{\,a^kb^kc^k\dollar_r\mid k\geq 0\,\}$ is
not context free. {F}rom the contradiction we derive that~$|R(m)|$ is
finite, for all $m\geq 0$.

In particular, this means that for every $m\geq 0$ there is at least one
$n\geq 0$ such that $M$ accepts the input $a^mb^mc^m \dollar_l c^nb^na^n$,
whereby the right head reaches the center marker not before the left head.
Based on this fact, we now can construct a nondeterministic pushdown
automaton $M_2$ that accepts the language
$L_2=\{\,a^jb^jc^j\dollar_l \mid j\geq 0\,\}$ as follows.

On a given input $a^jb^jc^j\dollar_l$, $M_2$ simulates a computation of
$M$ on input $a^jb^jc^j \dollar_l c^kb^la^m$, whereby $M_2$ guesses
the suffix $c^kb^la^m$ step-by-step.
So, whenever $M$ moves its left head to the right, $M_2$ moves its sole
head to the right as well. Whenever $M$ moves its right head to the left, 
$M_2$ guesses the next symbol from the suffix.
If $M_2$ guesses the center marker and, thus, in the simulation the
right head would see the center marker before the left head, $M_2$
rejects. If in the simulation the left head sees the center marker
before the right head, the simulation continues with $\lambda$-steps
until $M_2$ guesses that both heads meet.
In this case,~$M_2$ accepts if and only if $M$ accepts.
Considering the language accepted by $M_2$, one sees that
on some input $a^jb^jc^j \dollar_l$, $j\geq 0$, $M_2$ can guess a
suffix $c^nb^na^n$ such that $M$ accepts the input $a^jb^jc^j \dollar_l c^nb^na^n$,
whereby the right head reaches the center marker not before the left head
(it follows from above that such a suffix exists). So, $M_2$ accepts
any word of the form $a^jb^jc^j \dollar_l$, $j\geq 0$.
Conversely, if an input $w\dollar_l$ is not of the form $a^jb^jc^j \dollar_l$, then
there is no computation of $M_2$ that
accepts any word with prefix $w\dollar_l$
(due to the center marker $\dollar_l$, $w$ is verified to have form
$a^jb^jc^j$). Therefore, the simulation cannot ending accepting
and, thus, $M_2$ rejects.

So, since $M_2$ is a pushdown automaton, $L_2$ must
be a context-free language. {F}rom the contradiction we derive 
that~$L_{dta}$ is not accepted by any deterministic double-head 
pushdown automaton.
\end{proof}

The next example shows that the language $L_{dta}$ is accepted even by
input-driven double-head pushdown automata provided that nondeterminism is allowed.

\begin{example}\label{exa:ldta-accepted}
The language $L_{dta}=\{\, a^ib^jc^k \dollar_r c^nb^na^n \mid i,j,k,n\geq 0\,\}
\cup \{\, a^nb^nc^n \dollar_l c^ib^ja^k \mid i,j,k,n\geq 0\,\}$
is accepted by the $\niddhpda$ 
\[
M=\langle
\{s_0,p_1,p_2,\dots, p_7, q_1,q_2,\dots,q_7,s_+\},\Sigma_D\cup\Sigma_R\cup\Sigma_N,\{D,G,A\},s_0,
 \{p_4,q_4,s_+\}, \bot, \delta_D, \delta_R, \delta_N \rangle,
\]
where $\Sigma_D=\{a,\dollar_l,\dollar_r\}$, $\Sigma_R=\{b\}$,
$\Sigma_N=\{c\}$, and
the transition functions are defined as follows. 
In its first step, $M$ guesses whether the input contains a symbol
$\dollar_r$ or a symbol $\dollar_l$. Dependent on the guess one of the two
subsets in the definition of $L_{dta}$ are verified, and $M$ starts to read
the prefix from the left or the suffix from the right.  
The states $p_\ell$ are used
in the case of a $\dollar_r$, the states $q_\ell$ are used in the case of a
$\dollar_l$, and the remaining states are used for both cases. Recall that
$i,j,k$ and $n$ may be zero.
$$
\begin{array}[t]{rrcl}
 (1) & \delta_D(s_0,a,\lambda,\bot) &\ni& (p_1, D)\\
 (2) & \delta_R(s_0,b,\lambda,\bot) &\ni& p_2\\
 (3) & \delta_N(s_0,c,\lambda,\bot) &\ni& p_3\\
 (4) & \delta_D(s_0,\dollar_r,\lambda,\bot) &\ni& (p_4, G)\\[2mm]
\end{array}
\qquad\qquad
\begin{array}[t]{rrcl}
 (5) & \delta_D(s_0,\lambda,a,\bot) &\ni& (q_1, D)\\
 (6) & \delta_R(s_0,\lambda,b,\bot) &\ni& q_2\\
 (7) & \delta_N(s_0,\lambda,c,\bot) &\ni& q_3\\
 (8) & \delta_D(s_0,\lambda,\dollar_l,\bot) &\ni& (q_4, G)\\[2mm]
\end{array}
$$
We continue to construct the transition functions for the first case,
the construction for the second case is symmetric. So, while $M$ processes
the prefix $a^ib^jc^k$ it has to verify the form of the prefix and has to obey
the action associated to the symbols. The actual pushdown content generated
in this phase does not matter. Therefore, $M$ pushes dummy symbols $D$ and
a special symbol $G$ when it has reached the $\dollar_r$. The form of the
prefix is verified with the help of the states $p_1,p_2$, and $p_3$:
$$
\begin{array}[t]{rrcl}
 (9) & \delta_D(p_1,a,\lambda,D) &\ni& (p_1, D)\\
 (10) & \delta_R(p_1,b,\lambda,D) &\ni& p_2\\
 (11) & \delta_N(p_1,c,\lambda,D) &\ni& p_3\\
 (12) & \delta_D(p_1,\dollar_r,\lambda,D) &\ni& (p_4, G)\\[2mm]
\end{array}
\qquad\qquad
\begin{array}[t]{rrcl}
 (13) & \delta_R(p_2,b,\lambda,D) &\ni& p_2\\
 (14) & \delta_N(p_2,c,\lambda,D) &\ni& p_3\\
 (15) & \delta_D(p_2,\dollar_r,\lambda,D) &\ni& (p_4, G)\\[2mm]
 (16) & \delta_N(p_3,c,\lambda,D) &\ni& p_3\\
 (17) & \delta_D(p_3,\dollar_r,\lambda,D) &\ni& (p_4, G)\\[2mm]
\end{array}
$$
In the next phase, $M$ has to verify the suffix $c^nb^na^n$. To this end,
it moves its both heads alternately where for each but the first $a$ an
$A$ is pushed. For the first $a$ (if it exists) a symbol $G$ is pushed.
So, if both heads, one after the other, arrive at the $b$-sequence, the
number $n$ of $a$'s coincides with the number of $c$'s and the pushdown
content is of the form $A^{n-1}GGD^*$, if $n\geq 1$. If $n=0$, $M$ halts
in the accepting state~$p_4$. If $n\geq 1$, in the final phase the
$b$-sequence is read and its length is compared with the number of
$A$'s at the top of the stack. The final step reads the last $b$ 
and the $G$ from the pushdown store and enters the accepting state $s_+$,
for which no further transitions are defined. So, a correct input
is accepted and an incorrect input is not:
 $$
\begin{array}[t]{rrcl}
 (18) & \delta_D(p_4,\lambda,a,G) &\ni& (p_5, G)\\
 (19) & \delta_N(p_5,c,\lambda,A) &\ni& p_6\\
 (20) & \delta_N(p_5,c,\lambda,G) &\ni& p_6\\[2mm]

 (20) & \delta_D(p_6,\lambda,a,A) &\ni& (p_5,A)\\
 (21) & \delta_D(p_6,\lambda,a,G) &\ni& (p_5,A)\\
\end{array}
\qquad\qquad
\begin{array}[t]{rrcl}
 (20) & \delta_R(p_6,\lambda,b,A) &\ni& p_7\\
 (21) & \delta_R(p_6,\lambda,b,G) &\ni& s_+\\[2mm]

 (22) & \delta_R(p_7,\lambda,b,A) &\ni& p_7\\
 (23) & \delta_R(p_7,\lambda,b,G) &\ni& s_+\\
\end{array}
$$
\eoe
\end{example}

By Lemma~\ref{lem:lang-det-not} and Example~\ref{exa:ldta-accepted}
we conclude the next theorem.

\begin{theorem}\label{theo:ndet-better}
The family $\lfam(\diddhpda)$ is strictly included in the family
$\lfam(\niddhpda)$ and the family $\lfam(\ddoiddhpda)$ is strictly included in the family
$\lfam(\ndoiddhpda)$.
\end{theorem}

Next, we compare the computational capacities of $\iddhpda$s and
$\doiddhpda$s, that is, the capacity gained in providing (possibly
different) signatures to each of the heads. As witness language for 
the result that two signatures are better than one we use
the language $L_{ta}=\{\, a^nb^na^n \mid n\geq 1\,\}$.

\begin{example}\label{exa:dta-accepted-double}
The language $L_{ta}=\{\, a^nb^na^n \mid n\geq 1\,\}$
is accepted by the $\ddoiddhpda$ 
\[
M=\langle
\{q_0,q_1,q_2,q_a,q_+\},\Sigma_{D,l}\cup\Sigma_{D,r}\cup\Sigma_{R,l}\cup\Sigma_{R,r}\cup
\Sigma_{N,l}\cup\Sigma_{N,r},\{A,G\},q_0,
 \{q_+\}, \bot, \delta_D, \delta_R, \delta_N \rangle,
\]
where 
$\Sigma_{D,l}=\{a\}$, $\Sigma_{D,r}=\emptyset$,
$\Sigma_{R,l}=\{b\}$, $\Sigma_{R,r}=\{b\}$,
$\Sigma_{N,l}=\emptyset$, $\Sigma_{N,r}=\{a\}$,
and
the transition functions are defined as follows. 
In its first step, $M$ reads an $a$ with its left head and pushes 
the special symbol $G$ into the pushdown store. Subsequently, it reads
an $a$ with its right head while the pushdown store remains unchanged.
Next, $M$ moves its both heads alternately where for each $a$ read by
the left head (in state $q_1$) an $A$ is pushed and for each $a$ read
by the right head (in state $q_a$) the pushdown store remains unchanged.
Let $X\in\{A,G\}$.
$$
\begin{array}[t]{rrcl}
 (1) & \delta_D(q_0,a,\lambda,\bot) &=& (q_a, G)\\
 (2) & \delta_D(q_1,a,\lambda,X) &=& (q_a, A)\\[2mm]

 (3) & \delta_N(q_a,\lambda,a,X) &=& q_1\\
\end{array}
$$
So, after having read $n\geq 1$ symbols $a$ with the left as well as with
the right head, $M$ is in state $q_1$ and the pushdown store contains
the word $A^{n-1}G$. The next phase starts when the right head
reads a $b$ in state~$q_1$. In this phase, the right head is not used.
The left head reads the $b$'s while for each $b$ an $A$ is popped. When
$M$ reads a~$b$ with the special symbol $G$ on top of the
pushdown store, the sole accepting state~$q_+$ is entered.
$$
\begin{array}[t]{rrcl}
 (4) & \delta_R(q_1,b,\lambda,A) &=& q_2\\
 (5) & \delta_R(q_1,b,\lambda,G) &=& q_+\\[2mm]
 (6) & \delta_R(q_2,b,\lambda,A) &=& q_2\\
 (7) & \delta_R(q_2,b,\lambda,G) &=& q_+\\
\end{array}
$$
In this way, clearly, any word from $L_{ta}$ is accepted by $M$. Conversely,
in order to reach its sole accepting state $q_+$, $M$ has to read at least
one $a$ from the left as well as from the right, otherwise it cannot enter
state~$q_1$. Moreover, whenever it reads an $a$ from the left it must
read an $a$ from the right, otherwise it would halt in state $q_a$. 
Since
the transition functions are undefined for input symbol $b$ and
states~$q_0$ and~$q_a$,~$M$ cannot accept without having read a $b$. 
In order to enter the accepting state $q_+$ it must have read as many
$b$'s as $a$'s from the prefix which, in turn, have been read from
the suffix as well. Since the transition functions are undefined for
state $q_+$, $M$ necessarily halts when this state is entered.
If in this situation the input has been read entirely, $M$ accepts.
This implies that any word accepted by $M$ belongs to $L_{ta}$.
\eoe
\end{example}

\begin{lemma}\label{lem:lang-simple-not}
The language $L_{ta}$ is not accepted by any nondeterministic input-driven double-head
pushdown automaton.
\end{lemma}

By Example~\ref{exa:dta-accepted-double} and Lemma~\ref{lem:lang-simple-not}
we conclude the next theorem.

\begin{theorem}\label{theo:double-better}
The family $\lfam(\diddhpda)$ is strictly included in the family
$\lfam(\ddoiddhpda)$ and the family $\lfam(\niddhpda)$ is strictly included in the family
$\lfam(\ndoiddhpda)$.
\end{theorem}

So far, we have shown that nondeterminism is better than determinism
for a single as well as for double signatures, and that double signatures
are better than a single signature for deterministic as well as
nondeterministic computations. Moreover, Lemma~\ref{lem:lang-simple-not}
shows that language $L_{ta}$ is not accepted by any nondeterministic 
input-driven double-head pushdown automaton, while 
Example~\ref{exa:dta-accepted-double} shows that language~$L_{ta}$
is accepted, even by a deterministic input-driven double-head pushdown
automaton, provided that double signatures are available. Conversely,
Example~\ref{exa:ldta-accepted} reveals that language~$L_{dta}$
is accepted even with a nondeterministic input-driven double-head pushdown
automaton with a single signature, while Lemma~\ref{lem:lang-det-not}
shows that, to this end, determinism is not sufficient. This implies
the next corollary.

\begin{corollary}\label{cor:det-double-vs-ndet-single}
The families $\lfam(\ddoiddhpda)$ and $\lfam(\niddhpda)$
are incomparable.
\end{corollary}

Next we turn to compare the
four language families under consideration with some other well-known language
families.

A context-sensitive grammar is said to be \emph{growing context sensitive}
if the right-hand side of every production is \emph{strictly} longer 
than the left-hand side. The family of growing context-sensitive languages ($\textrm{GCSL}$)
lies strictly in between the context-free and context-sensitive languages.
Another language family lying properly in between
the regular and the growing context-sensitive languages are the \emph{Church-Rosser languages}
($\textrm{CRL}$), which have been introduced in~\cite{McNaughton:1988:crtsfl}. They
are defined via finite, confluent, and length-reducing Thue
systems. Church-Rosser languages are incomparable to the context-free
languages~\cite{Buntrock:1998:gcslcrl:art} and have neat properties.  For
example, they parse rapidly in linear time, contain non-semilinear as well as
inherently ambiguous languages~\cite{McNaughton:1988:crtsfl}, are
characterized by deterministic automata
models~\cite{Buntrock:1998:gcslcrl:art,Niemann:2005:crldvgcsl}, and contain
the deterministic context-free languages ($\textrm{DCFL}$) as well as their reversals
($\textrm{DCFL}^R$) properly~\cite{McNaughton:1988:crtsfl}.

\begin{theorem}\label{theo:compare-better}
Each of the families $\lfam(\diddhpda)$, $\lfam(\niddhpda)$,
$\lfam(\ddoiddhpda)$, and $\lfam(\ndoiddhpda)$ is incomparable
with $\textrm{GCSL}$ as well as with $\textrm{CRL}$.
\end{theorem}

\begin{proof}
Example~\ref{exa:gladkij} shows that the marked Gladkij language 
$
L_1=\{\,w\border h_1(w)^R \border h_2(w) \mid w \in \{a,b\}^*\,\},
$
where the homomorphisms $h_1$ and $h_2$ are defined by
$h_1(a)=\bar{a}$, $h_1(b)=\bar{b}$, $h_2(a)=\hat{a}$, and $h_2(b)=\hat{b}$,
is accepted by some $\diddhpda$. Language $L_1$ is not growing context
sensitive and, thus, is not a Church-Rosser language~\cite{Buntrock:1998:gcslcrl:art}.

Conversely, the unary language $\{\,a^{2^n} \mid n \ge 0\,\}$ is not
semilinear, but a Church-Rosser language~\cite{McNaughton:1988:crtsfl}.
Since every language accepted even by some nondeterministic double-head pushdown
is semilinear~\cite{Nagy:2015:afothpa} the incomparabilities claimed follow.
\end{proof}

The inclusion structure of the families in question is depicted in
Figure~\ref{fig:fam-inclusions}.

\begin{figure}[!ht]
\begin{center}
\includegraphics[scale=.8]{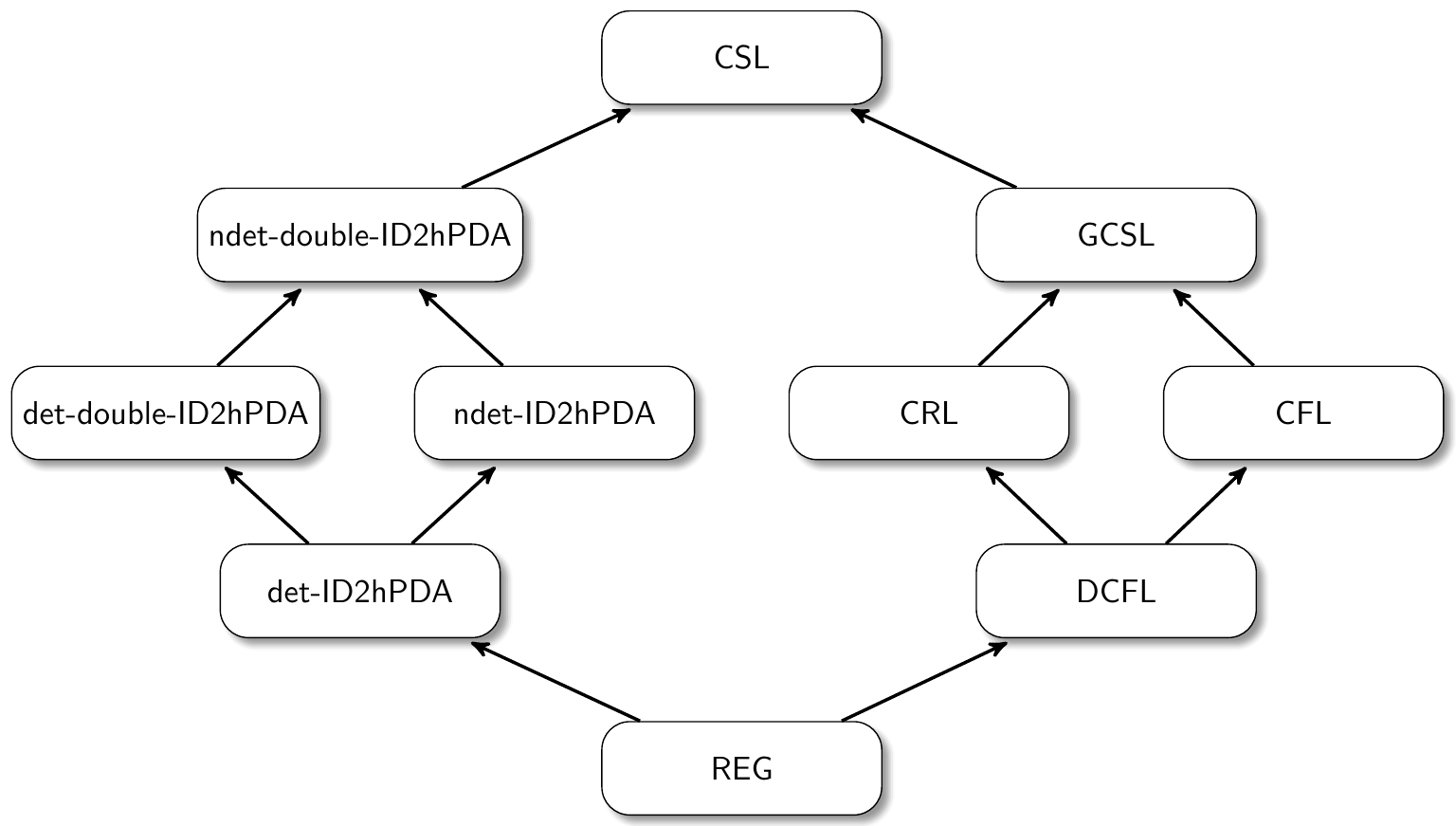}%
\end{center}
\caption{Inclusion structure of language families. The arrows indicate strict
inclusions. All nodes not connected by a directed path are incomparable, where
the incomparability with the (deterministic) context-free languages is a
conjecture.
}\label{fig:fam-inclusions}
\end{figure} 

\section{Closure properties}
\label{sec:closure}

We investigate the closure properties of the language families induced
by nondeterministic and deterministic input-driven $\dhpda$s.
Table~\ref{tab:closure} summarizes our results.
\begin{table}[!tbh]
\renewcommand{\arraystretch}{1.15}\setlength{\tabcolsep}{5pt}
\begin{center}
\begin{tabular}{|c|c|c|c|c|c|c|c|c|c|c|}
\cline{2-11}
\multicolumn{1}{c|}{\rule{0pt}{10pt}}
 & $\overline{\phantom{aa}}$ & $\cup$& $\cap$ & ${}\cup \textrm{REG}$& ${}\cap 
\textrm{REG}$ 
& $\cdot$& $*$ & $h$ & $h^{-1}$ & $R$\\
\hline
$\lfam{(\dhpda)}$ & no & yes & no & yes & yes & no & no & yes & ? & yes\\
$\lfam{(\diddhpda)}$  & yes & no & no & yes & yes & no & no & no & no &  yes\\
$\lfam( \niddhpda)$   & no & no & no & yes & yes & no & no & no & no &yes\\ 
\hline
\end{tabular}
\end{center}
\caption{Closure properties of the language classes in question.}
\label{tab:closure}
\end{table}

In \cite{Nagy:2015:afothpa} it was shown that the family of languages
accepted by ordinary nondeterministic double-head pushdown automata is
closed under union, homomorphism, and reversal, but it is \emph{not}
closed under intersection, complement, concatenation, and
iteration. Furthermore it was shown that the languages 
$$L_1 =\{\,a^nb^nc^nd^ne^n\mid n\ge 1\,\},\quad 
L_2^2,\quad{and}\quad L_2^*,\quad\mbox{for $L_2=\{\,a^nb^nc^nd^n\mid
  n\ge 1\,\}$}$$ cannot be accepted by any nondeterministic
double-head pushdown automaton. We start our investigation with the
language families $\lfam{(\diddhpda)}$ and $\lfam(\niddhpda)$.

Our first result is an easy observation that if the input heads of a
(input-driven) \dhpda\ change their roles, the accepted language is
the reversal of the original language. Hence we show closure under
reversal of the language families in question.

\begin{theorem}\label{theo:closed-reversal}
  Both language families $\lfam(\diddhpda)$ and $\lfam(\niddhpda)$ are
  closed under reversal.
\end{theorem}

\begin{proof}
Let $M=\langle Q,\Sigma, \Gamma, q_0, F,\bot, 
\delta_D, \delta_R, \delta_N\rangle$ be an 
$\niddhpda$. We construct an 
$\niddhpda$
$M'=\langle Q,\Sigma, \Gamma, q_0, F,\bot, 
\delta_D', \delta_R', \delta_N'\rangle$,
where the transition function is defined as follows:
for every $q,q'\in Q$,\ $z,z'\in\Gamma^*$, and $x,y\in\Sigma$ we set 
\begin{itemize}
\item $(q',z') \in \delta_D'(q,y,x,z)$, if $(q',z') \in
  \delta_D(q,x,y,z)$,
\item $q' \in \delta_R'(q,y,x,z)$, if
  $q' \in \delta_R(q,x,y,z)$, and
\item $q' \in \delta_N'(q,y,x,z)$,
  if $q' \in \delta_N(q,x,y,z)$.
 \end{itemize}
 Then it is easy to see by induction on the length of the computation
 that~$M'$ accepts the reversal of the language~$L(M)$, that is,
 $L(M')=L(M)^R$. Observe, that~$M'$ is deterministic, if~$M$
 was. Thus, we have shown closure of both language families under the
 reversal operation.
\end{proof}

The above mentioned result that the family of languages accepted by
double-head pushdown automata is not closed under intersection carries
over to the input-driven case as well.

\begin{theorem}\label{theo:nclosed-intersection}
Both families $\lfam(\diddhpda)$ and 
$\lfam(\niddhpda)$ are \emph{not} closed under intersection.
\end{theorem}

\begin{proof}
  In~\cite{Nagy:2015:afothpa} the non-closure of $\lfam(\dhpda)$ under
  intersection was shown with the help of the \dhpda\ languages
  $L=\{\,a^nb^nc^nd^me^\ell\mid m,n,\ell\ge 1\,\}$ and $L'=\{
  a^mb^\ell c^n d^ne^n\mid n,m,\ell\geq 1\,\}$, since their
  intersection $L\cap L'=\{\,a^nb^nc^nd^ne^n\mid n\geq 1\,\}$ is not
  member of $\lfam(\dhpda)$. Thus, in order to prove our non-closure
  result on intersection it suffices to show that both languages~$L$
  and~$L'$ can already be accepted by a \diddhpda.

  We only give a brief description of a \diddhpda~$M$ that accepts the
  language~$L$. By a similar argumentation one can construct a
  \diddhpda\ for the language~$L'$, too. On input~$w$ the
  \diddhpda~$M$ proceeds as follows: the right head of $M$ moves from
  right to left until it reaches the first~$c$ and checks whether the
  input has a suffix of the~$d^me^\ell$, for some $m,\ell\ge 1$.  This
  can be done without using the pushdown store and without moving the
  left head.  Afterwards it again moves only its right head and pushes
  a~$C$ for every letter~$c$ into the pushdown store. When it reaches the
  first $b$ it starts alternately moving the left and the right head,
  reading letter~$b$ from the right and~$a$ from the left, beginning
  with the right head, while it pops for every movement of the right
  head a~$C$ from the pushdown store. If the pushdown store is empty
  and the left head moves to the right, it enters an accepting
  state. The alphabets of $M$ are $\Sigma_D=\{c\}, \Sigma_R=\{b\},
  \Sigma_N=\{a,d,e\}$. A detailed construction of~$M$ is left to the
  reader. This proves the stated claim.
\end{proof}

Before we continue with the complementation operation, we first
establish that every deterministic and nondeterministic input-driven
double-head pushdown automaton can be forced to read the entire input.
This property turns out to be useful for the following construction
showing the closure under complementation for deterministic input-driven 
double-head pushdown automata.

\begin{lemma}\label{lemm:complete-input}
  Let $M$ be an \niddhpda. Then one can construct an equivalent
  \niddhpda~$M'$, that is, $L(M')=L(M)$, that decides on
  acceptance/rejection after it has read the entire input. If~$M$ is
  deterministic, then so is~$M'$.
\end{lemma}

The family of languages accepted by double-head pushdown automata are
not closed under complementation. We show that the family of languages
accepted by deterministic input-driven double-head pushdown automata
is closed under complementation, while the nondeterministic family is
not closed.

\begin{theorem}\label{theo:clos-comp}
The family $\lfam(\diddhpda)$ is closed under complementation.
\end{theorem}

\begin{proof}
  Let $M=\langle Q,\Sigma, \Gamma, q_0, F,\bot, \delta_D, \delta_R,
  \delta_N\rangle$ be a \diddhpda. By the previous lemma we can assume
  w.l.o.g.\ that~$M$ decides on acceptance/rejection after it has read
  the entire input. But then, if we exchange accepting and
  non-accepting states we accept the complement of~$L(M)$. Thus, 
  the \diddhpda\ $M'=\langle Q,\Sigma, \Gamma, q_0, F',\bot, \delta_D,
  \delta_R, \delta_N\rangle$ with $F'=Q\setminus F$ is an acceptor for the 
  language $\overline{L(M)}$. This proves our statement.
\end{proof}

Since the family of languages accepted by $\diddhpda$ is not closed
under intersection, it can be concluded that it is not closed under
union.

\begin{theorem}\label{theo:det-closed-union}
The family $\lfam(\diddhpda)$ is not closed under union.
\end{theorem}

Let us come back to the complementation operation. For the language
family induced by \niddhpda\ we obtain non-closure under
complementation in contrast to the above given theorem on the
deterministic language family in question.

\begin{theorem}\label{theo:ndet-not-closed-int}
The family $\lfam(\niddhpda)$ is not closed under complementation.
\end{theorem}

\begin{proof}
  In \cite{Nagy:2015:afothpa} it has been shown that the language
  $L_1=\{ a^nb^nc^nd^ne^n\mid n\ge 1\}$ cannot be accepted even by any
  double-head pushdown automata. We briefly show that the language
  $\overline{L_1}$ is accepted by an \niddhpda~$M$. The complement
  of~$L_1$ can be described as follows: a word is in~$\overline{L_1}$
  if and only if (i) it belongs to \emph{complement} of the regular
  language $a^+b^+c^+d^+e^+$ or (ii) it belongs to one of the
  context-free languages $\{\,a^{n_1}b^{n_2}c^{n_3}d^{n_4}e^{n_5}\mid
  \mbox{$n_1,n_2,\ldots, n_5\geq 1$ and $n_i\neq n_j$}\,\}$, for some
  pair $(i,j)\in\{1,2,\ldots,5\}^2$ with $i\neq j$. Thus, on input~$w$
  the \niddhpda~$M$ guesses which of the above cases~(i) or~(ii)
  applies. In the first case~$M$ simulates a finite automaton without
  using its pushdown store. In the second case, automaton~$M$ guesses
  appropriate~$i$ and $j$ with $i\neq j$ from $\{1,2,\ldots, 5\}$ and
  moves its two heads to the corresponding blocks of letters. Then it
  checks whether $n_i\neq n_j$ by alternately moving the left and right
  head without using the pushdown store. If $n_i\neq n_j$ and the
  heads meet, the automaton accepts, otherwise it rejects. Since~$M$
  is not using the pushdown store at all, only the transition
  function~$\delta_N$ is defined. Thus, the signature is
  $\Sigma_N=\Sigma$ and $\Sigma_D=\Sigma_R=\emptyset$. 

  Since~$M$ accepts~$\overline{L_1}$, but $L_1$ cannot be accepted by
  any double-head pushdown automata, the language family
  $\lfam(\niddhpda)$ is not closed under complementation.
\end{proof}

Now, $L_2$ can be used to show that the family of languages 
accepted by deterministic and nondeterministic input-driven double-head 
pushdown automata are not closed under concatenation and iteration.

\begin{theorem}\label{theo:ndet-not-closed-concat}
  Both language families $\lfam(\diddhpda)$ $\lfam(\niddhpda)$ are
  \emph{not} closed under concatenation and iteration.
\end{theorem}

While both families $\lfam(\diddhpda)$ and $\lfam(\niddhpda)$ are not
closed under union and intersection, they are closed under the union
and intersection with regular languages.

\begin{theorem}\label{theo:ndet-closed-regunion}
  Both families $\lfam(\diddhpda)$ and $\lfam(\niddhpda)$ are closed
  under intersection and union with regular languages.
\end{theorem}

Next, we consider the closure under homomorphism.

\begin{theorem}\label{theo:ndet-not-hom}
  Both families $\lfam(\diddhpda)$ and $\lfam(\niddhpda)$ are
  \emph{not} closed under (length preserving) homomorphisms.
\end{theorem}

\begin{proof}
  Consider the language $L=\{\,a^n\dollar b^n\dollar c^n\mid n\ge
  1\,\}$. It is easy to show that~$L$ is accepted by some \diddhpda\
  with signature $\Sigma_N=\{a,\dollar\}$, $\Sigma_D=\{c\}$, and
  $\Sigma_R=\{b\}$. The details are left to the reader. Further
  consider the homomorphism~$h$ defined by $h(a)=a$,\
  $h(\dollar)=\dollar$,\ $h(b)=a$, and $h(c)=a$ that leads to the
  language $h(L)=\{ a^n\dollar a^n\dollar a^n\mid n\ge 1\}$. We show
  that $h(L)$ cannot be accepted by any \niddhpda.

  Assume to the contrary that there is an \niddhpda~$M$ accepting the
  language~$h(L)$. Observe, that the input contains only the
  letters~$a$ and two~$\dollar$s. Then we consider three cases,
  according to which set of the signature the letter~$a$ belongs to:
  \begin{enumerate}
  \item Letter $a\in\Sigma_N$. Then~$M$ possibly can use the pushdown
    store only for the two~$\dollar$ letters. In this case, the whole
    computation of~$M$ can be mimicked by a finite automaton.
  \item Letter $a\in\Sigma_D$. Then the pushdown of~$M$ can
    arbitrarily increase in height during a computation (if the
    $a$-blocks on both sides of the word are long enough), and can be
    decreased at most twice with the help of the
    letters~$\dollar$. Again, the whole computation of~$M$ can be
    simulated by a finite automaton.
  \item Letter $a\in\Sigma_R$. Again, the computation of~$M$ can be
    done by a finite state machine, since the pushdown height is
    bounded by two and can be stored in the finite control of an
    automaton. The letters~$a$ force a pop and the letters~$\dollar$
    may increase the pushdown height by at most two.
  \end{enumerate}
  By our consideration we conclude that~$M$ can always be replaced by
  a finite state automaton and therefore the language~$h(L)$ is
  regular, which is a contradiction to the pumping lemma of regular
  languages and to our above given assumption. Hence~$h(L)$ cannot be
  accepted by any \niddhpda.
\end{proof}

From the Boolean operations, the union operation applied to
$\lfam(\niddhpda)$ is still missing.

\begin{theorem}\label{theo:ndet-not-union}
The family 
$\lfam(\niddhpda)$ is \emph{not} closed under union.
\end{theorem}

\begin{proof}
  Consider the language $L=\{\,a^nb^{2n}a^n\mid n\ge 1\,\}$. Using the
  signature $\Sigma_N=\emptyset$,\ $\Sigma_D=\{a\}$, and
  $\Sigma_R=\{b\}$ it is not hard to see that~$L$ is accepted by a
  \diddhpda. Note that the $a^n$-prefix and -suffix of the input word
  is compared by the use of the two input heads, while the
  $b^{2n}$-infix is checked against the content of the pushdown, which
  is previously filled by reading the $a$'s from the input. Similarly,
  the language where the~$a$'s and~$b$'s are exchanged, that is,
  $L'=\{\,b^na^{2n}b^n\mid n\geq \,\}$ is also accepted by some \diddhpda.
 
  Next, we consider the union $L\cup L'$. We show that it cannot be
  accepted by any \niddhpda. Assume to the contrary that there is an
  \niddhpda~$M''=\langle Q'',\{a,b\},\Gamma,
  q_,F,\bot,\delta_D,\delta_R,\delta_n\rangle$ that accepts the
  language $L\cup L'$, that is, $L(M'')=L\cup L'$.  With the same
  argumentation as in the proof of Theorem~\ref{theo:ndet-not-hom} we
  conclude that~$M''$ cannot accept $L\cup L'$ without using the
  pushdown store. Thus, one of the two input symbols~$a$ or~$b$
  force~$M''$ to push and the other symbol to pop. W.l.o.g. we assume
  that~$\Sigma_D=\{a\}$ and $\Sigma_R=\{b\}$; the other case can be
  treated in a similar way.
  Recall that~$Q''$ is the state set of~$M''$. Then consider the input
  word $w=b^na^{2n}b^n$, for $n>|Q''|$. Note that while the input heads
  cross over the $b^n$-prefix and -suffix of~$w$ the automaton~$M''$
  is forced to pop from the pushdown store and thus empties it.
  Since $w\in L\cup L'$, there is an accepting computation of~$M''$
  on~$w$, which is of the form
  $$(q_0,b^na^{2n}b^n,\lambda)\vdash^*(s,b^{n-i}a^{2n}b^{n-j},\lambda)\vdash^*
  (s,b^{n-i-i'}a^{2n}b^{n-j-j'},\lambda)\vdash^* (q_f,\lambda,\gamma),$$
  where $s\in Q''$,\ $q_f\in F$, $\gamma\in\Gamma^*$, and moreover,
  $i+i'\leq n$ and $j+j'\leq n$ and $i'+j'\geq 1$. But then by cutting
  out the loop computation on the state~$s$ also the word
  $b^{n-i'}a^{2n}b^{n-j'}$ is accepted by~$M''$ \textit{via} the
  computation
  $$(q_0,b^{n-i'}a^{2n}b^{n-j'},\lambda)\vdash^*
  (s,b^{n-i-i'}a^{2n}b^{n-j-j'},\lambda)\vdash^* (q_f,\lambda,\gamma).$$
  Since this word is not a member of $L\cup L'$ we get a contradiction
  to our assumption. Therefore, the language~$L\cup L'$ cannot be accepted
  by any \niddhpda. Thus, the language family $\lfam(\niddhpda)$
  is not closed under union.
\end{proof}

For the inverse homomorphism we also get a non-closure result.

\begin{theorem}\label{theo:ndet-not-closed-inverse}
  Both families $\lfam(\diddhpda)$ and $\lfam(\niddhpda)$ are
  \emph{not} closed under inverse homomorphisms.
\end{theorem}

\begin{proof}
  Consider the \diddhpda\ language $L=\{ a^nb^{2n}a^n \mid n\ge 1\}$
  from Lemma~\ref{theo:ndet-not-union}. Let~$h$ be the homomorphism
  defined by $h(a)=a$ and $h(b)=bb$. Then $h^{-1}(L)=\{\,a^nb^na^n\mid
  n\geq 1\,\}$. This is the language~$L_{ta}$ from
  Lemma~\ref{lem:lang-simple-not} that cannot be accepted by any
  \niddhpda. This shows that both language families
  $\lfam(\diddhpda)$ and $\lfam(\niddhpda)$ are not closed under
  inverse homomorphisms.
\end{proof}

\section{Decidability questions}
\label{sec:dec}

In this section, we investigate the usually studied decidability questions for deterministic and nondeterministic
input-driven $\dhpda$s.
It turns out that the results are similar to those obtained for conventional deterministic and nondeterministic pushdown automata.
In particular, we obtain the decidability of emptiness and finiteness for $\diddhpda$s and $\niddhpda$s
as well as the decidability of equivalence with a regular set, inclusion in a regular set, and inclusion of a regular set.
On the other hand, inclusion turns out to be not even semidecidable for $\diddhpda$s and universality, equivalence, and regularity
are not semidecidable for $\niddhpda$ as well. Finally, the decidability and non-semidecidability results can be translated
to hold for $\doiddhpda$s correspondingly.

\begin{theorem}\label{thm:dec:emptiness}
Let $M$ be an $\ndhpda$. Then, it is decidable whether or not $L(M)$ is empty or finite.
\end{theorem}

\begin{proof}
  In~\cite{Nagy:2015:afothpa} Nagy shows that for every $\ndhpda$ a
  classical $\npda$ accepting a letter-equivalent context-free
  language can effectively be constructed.  Thus, the emptiness and
  finiteness problems for an $\ndhpda$ can be reduced to the
  corresponding problems for an $\npda$ which are known to be
  decidable.
\end{proof}

\begin{corollary}\label{cor:dec:emptiness}
Let $M$ be an $\niddhpda$ or $\diddhpda$. Then, it is decidable whether or not $L(M)$ is empty or finite.
\end{corollary}

To obtain undecidability results we will use the technique of \emph{valid computations} of Turing machines which
is presented, for example, in~\cite{Hopcroft:1979:itatlc:book}.
This technique allows to show that some questions are not only undecidable, but moreover not semidecidable,
where we say that a problem is \emph{semidecidable}
if and only if the set of all instances for which the answer is ``yes'' is recursively enumerable (see, for example,~\cite{Hopcroft:1979:itatlc:book}).
Let $M=\langle Q,\Sigma,T,\delta,q_0,B,F \rangle$ be a deterministic Turing machine, 
where $T$ is the set of tape symbols including the set of input symbols~$\Sigma$ and the blank symbol $B$, 
$Q$ is the finite set of states and $F\subseteq Q$ is the set of final states.  
The initial state is $q_{0}$ and $\delta$ is the transition function. 
Without loss of generality, we assume that Turing machines 
can halt only after an odd number of moves, halt whenever they enter an
accepting state, make at least three
moves, and cannot print blanks.
At any instant during a computation,~$M$ can be completely described by an \emph{instantaneous description} (ID) which is a string 
$t q t'\in T^{*}Q T^{*}$ with the following meaning: $M$ is in the
state~$q$, the non-blank tape content is the string $tt'$, and the head scans the 
first symbol of $t'$. The initial ID of $M$ on input $x\in\Sigma^{*}$ is $w_0=q_{0}x$. 
An ID is accepting whenever it belongs to $T^{*}FT^{*}$.
The set $\valc(M)$ of valid (accepting) computations of~$M$ consists of all
finite strings
of the form 
$
w_0\border w_2\border \cdots \border w_{2n} \dollar w_{2n+1}^R\border \cdots \border w^R_{3}\border w^R_{1}
$
where $\border,\dollar \notin T \cup Q$,
$w_i$, $0\leq i\leq 2n+1$, are instantaneous description of~$M$, 
$w_0$ is an initial ID,
$w_{2n+1}$ is an accepting (hence halting) configuration,  
$w_{i+1}$ is the successor configuration of~$w_i$, $0\leq i\leq 2n$.
The set of \emph{invalid computations} $\invalc(M)$ is the
complement of~$\valc(M)$ with respect to the alphabet 
$T \cup Q \cup \{\border,\dollar\}$. 

\begin{theorem}
Let $M_1$ and $M_2$ be two $\diddhpda$s. Then, the question $L(M_1) \cap L(M_2)=\emptyset$ is not semidecidable.
\end{theorem}

\begin{proof}
We will first show that the set of valid computation $\valc(M)$ of a Turing machine $M$ 
is the intersection of two languages $L_1$ and $L_2$ where each language
is accepted by some $\diddhpda$. We define $L_1$ to consist of all strings of the form
$
w_0\border w_2\border \cdots \border w_{2n} \dollar w_{2n+1}^R\border \cdots \border w^R_{3}\border w^R_{1}
$,
where $w_{i+1}$ is the successor configuration of~$w_i$ for all even $0\leq i\leq 2n$.
Language $L_2$ is defined as the set of all strings of the form 
$
w_0\border w_2\border \cdots \border w_{2n} \dollar w_{2n+1}^R\border \cdots \border w^R_{3}\border w^R_{1}
$,
where $w_{i+1}$ is the successor configuration of~$w_i$ for all odd $1\leq i\leq 2n-1$. Moreover,
$w_0$ is an initial ID and $w_{2n+1}$ is an accepting ID. It is clear that $L_1 \cap L_2 = \valc(M)$.
Next, we sketch how $L_1$ can be accepted by some $\diddhpda$ $M_1$. The partition of the input alphabet
is $\Sigma_D=\Sigma_R=\emptyset$ and $\Sigma_N=T \cup Q \cup \{\border,\dollar\}$. Thus, we will not
make use of the pushdown store in our construction. The basic idea is that both heads of~$M_1$ move
successively to the right resp. left checking that the ID that is seen by the right head is indeed
the successor configuration seen by the left head. This is possible since the changes between a 
configuration and its successor configuration are only local and hence can be checked using the state
set of~$M_1$. Moreover, the state set is also used to check the correct format of the input, where
the left head checks the input part to the left of the marker $\dollar$, whereas the right head checks
the input part to the right of $\dollar$. The computation ends accepting when both heads meet at
the marker $\dollar$ and all previous checks have been successful. A $\diddhpda$ $M_2$ for $L_2$ works
similarly. First, the left head has to skip the initial ID $w_0$. Then, both heads of~$M_2$ move
successively to the right resp. left checking that the ID that is seen by the left head is indeed
the successor configuration seen by the right head. Again, the correct format of the input is implicitly
checked. When the left head has reached the marker~$\dollar$,
the right head has to skip the accepting ID $w_{2n+1}$ and the computation ends accepting when both heads meet at
the marker $\dollar$ and all previous checks have been successful.
Since $L_1 \cap L_2 = \valc(M)$ and the emptiness problem for Turing machines is not semidecidable
(see, for example,~\cite{Hopcroft:1979:itatlc:book}), the claim of the theorem follows.
\end{proof}

Since $\lfam(\diddhpda)$ is closed under complementation owing to Theorem~\ref{theo:clos-comp}, we immediately
obtain that the inclusion problem is not semidecidable.

\begin{corollary}
Let $M_1$ and $M_2$ be two $\diddhpda$s. Then, it is not semidecidable whether or not $L(M_1) \subseteq L(M_2)$.
\end{corollary}

However, in case of regular languages we can decide inclusion and equivalence.

\begin{theorem}
Let $M$ be a $\diddhpda$ and $R$ be a regular language. Then, it is decidable whether or not $L(M)=R$, $R \subseteq L(M)$,
or $L(M) \subseteq R$.
\end{theorem}

\begin{proof}
First, we note that $R \subseteq L(M)$ if and only if $R \cap \overline{L(M)}=\emptyset$
and that $L(M) \subseteq R$ if and only if $L(M) \cap \overline{R}=\emptyset$. Since 
$\lfam(\diddhpda)$ is closed under complementation and under intersection with regular languages
by Theorem~\ref{theo:clos-comp} and Theorem~\ref{theo:ndet-closed-regunion},
the regular languages are closed under complementation, and emptiness is decidable for $\diddhpda$s owing
to Theorem~\ref{thm:dec:emptiness}, all claims of the theorem follow.
\end{proof}

The decidability of the latter questions gets lost if the given $\iddhpda$ is nondeterministic, since in this case
even the universality question is not semidecidable.

\begin{theorem}\label{thm:undec:ndet}
Let $M$ be an $\niddhpda$. Then, the questions of universality, equivalence, and regularity are not semidecidable.
\end{theorem}

Owing to Theorem~\ref{thm:dec:emptiness} it is clear that emptiness and finiteness are decidable for $\ddoiddhpda$s and $\ndoiddhpda$s as well.
Since the language family accepted by $\ddoiddhpda$s is also closed under complementation and intersection with regular languages, we obtain
that decidable questions for $\diddhpda$s are also decidable for $\ddoiddhpda$s. On the other hand, the non-semidecidability results
obtained for $\iddhpda$s in the single mode obviously hold for the double mode as well.
It is currently an open problem whether equivalence and regularity are decidable for $\diddhpda$s or $\ddoiddhpda$s,
whereas both problems are known to be decidable for $\dpda$s.

\section*{Acknowledgment}
We would like to thank D\'{a}vid Angyal who brought double-head pushdown
automata close to us while his visit of our institute. He participated 
in the discussions and his ideas were significant contributions to this paper.
We consider him truly a co-author, but he insisted not to put his name on the list.

\bibliographystyle{eptcs}
\bibliography{idpda}

\end{document}